\documentclass[onecolumn,10.5pt]{IEEEtran}

\usepackage{amsfonts,amsmath,amssymb,amsthm,caption}
\usepackage{amsbsy}
\usepackage{graphicx,multirow,bm}
\usepackage{color}
\usepackage{braket}
\usepackage{geometry} 
\makeatletter
\theoremstyle{mythm}
\newtheorem{theorem}{Theorem}
\newtheorem{example}[theorem]{Example}
\newtheorem{definition}[theorem]{Definition}
\newtheorem{remark}[theorem]{Remark}

\newtheorem{corollary}[theorem]{Corollary}
\newtheorem{lemma}[theorem]{Lemma}

\newtheorem{problem}[theorem]{Problem}

\allowdisplaybreaks[4]

\begin{document}

\title{ 
On the Information-theoretic Security of Combinatorial All-or-nothing Transforms
\vskip 0.5cm
}
\author{
 Yujie Gu,  
 Sonata Akao,
 Navid Nasr Esfahani, 
 Ying Miao,
 and Kouichi Sakurai 
\thanks{Y. Gu, S. Akao and K. Sakurai are with the Graduate School and Faculty of Information Science and Electrical Engineering, Kyushu University, Fukuoka, Japan. (e-mails: gu@inf.kyushu-u.ac.jp, akao.sonata.598@s.kyushu-u.ac.jp, sakurai@inf.kyushu-u.ac.jp)}
\thanks{N. Esfahani is with the David R. Cheriton School of Computer Science, University of Waterloo, Waterloo, Ontario, N2L 3G1, Canada. (e-mail: nnasresf@uwaterloo.ca)}
\thanks{Y. Miao is with the Faculty of Engineering, Information and Systems, University of Tsukuba, Tsukuba, Ibaraki 305-8573, Japan. (e-mail: miao@sk.tsukuba.ac.jp)}
\thanks{This work has been supported by JSPS Grant-in-Aid for Early-Career Scientists No. 21K13830 and JSPS Grant-in-Aid for Scientific Research (B) No. 18H01133.
}\\[0.5cm]
}

\maketitle

\begin{abstract}
All-or-nothing transforms (AONT) were proposed by Rivest as a message preprocessing technique for encrypting data to protect against brute-force attacks, and have numerous applications in cryptography and information security.  
Later the unconditionally secure AONT and their combinatorial characterization were introduced by Stinson. 
Informally, a combinatorial AONT is an array with the unbiased requirements and its security properties in general depend on the prior probability distribution on the inputs $s$-tuples. 
Recently, it was shown by Esfahani and Stinson that a combinatorial AONT has perfect security provided that all the inputs $s$-tuples are equiprobable, and has weak security provided that all the inputs $s$-tuples are with non-zero probability. 

This paper aims to explore on the gap between perfect security and weak security for combinatorial $(t,s,v)$-AONTs. 
Concretely, 
we consider the typical scenario that all the $s$ inputs take values independently (but not necessarily identically) and quantify the amount of information $H(\mathcal{X}|\mathcal{Y})$ about any $t$ inputs $\mathcal{X}$ that is not revealed by any $s-t$ outputs $\mathcal{Y}$. 
In particular, we establish the general lower and upper bounds on $H(\mathcal{X}|\mathcal{Y})$ for combinatorial AONTs using information-theoretic techniques, and also show that the derived bounds can be attained in certain cases. 
Furthermore, 
the discussions are extended for the security properties of combinatorial asymmetric AONTs. 
\end{abstract}

\section{Introduction}
\label{sec:Introduction}

The concept of an all-or-nothing transform (AONT) was introduced by Rivest \cite{Rivest1997}, as a strongly non-separable mode of operation, that is a preprocessing step prior to encryption such that missing any cipher-block prevents the attacker from obtaining information about the message-blocks. The original motivation behind AONTs was to impede brute-force attacks on block ciphers when the key length cannot be increased \cite{Rivest1997}. Since then, numerous applications and extensions of AONT have been studied and introduced within different context, e.g., cryptography, information security, and combinatorics \cite{Boyko1999, CDHKS2000, Desai2000, DSS2001, E2021, KSLC2019, SNS}. 
Informally, an AONT is an unkeyed, invertible transformation which maps a sequence of inputs $(x_1,x_2,\ldots,x_s)$ to a sequence of outputs $(y_1,y_2,\ldots,y_s)$ with the following properties:
\begin{enumerate}
    \item[(i)]
    given all $(y_1,y_2,\ldots,y_s)$, it is easy to compute $(x_1,x_2,\ldots,x_s)$; 
    \item[(ii)] 
    if any one of the $y_j$ is missing, then it is computationally infeasible to obtain any information about any $x_i$. 
\end{enumerate}
In contrast to the above computationally secure AONT, Stinson~\cite{Stinson2001} introduced the unconditionally secure AONT, which later was extended to the general scenario~\cite{DES2016,EGS2018} where more than one $y_j$ could be missing. Here we expose the definition of unconditionally secure AONT in terms of the entropy function $H(\cdot)$ in~\cite{DES2016}.

\begin{definition}\label{def-AONT}\rm
Let $X_1,\ldots, X_s$ and $Y_1, \ldots, Y_s$ be input and output random variables respectively, which take values from the finite set $\Gamma$ of size $v$. 
These $2s$ random variables define a \textit{$(t, s, v)$-AONT} provided that the following conditions are satisfied: 
\begin{enumerate}
    \item[1)] $H(Y_1,\ldots,Y_s| X_1,\ldots,X_s) =0$;
    \item[2)] $H(X_1,\ldots,X_s| Y_1,\ldots,Y_s) =0$;
    \item[3)] For all $\mathcal{X} \subseteq \{X_1,\ldots,X_s\}$ with $|\mathcal{X}|=t$, and for all $\mathcal{Y} \subseteq \{Y_1,\ldots,Y_s\}$ with $|\mathcal{Y}|=s-t$, it holds that 
    \begin{align}\label{eq-perfect-security}
    H(\mathcal{X}|\mathcal{Y}) = H(\mathcal{X}).    
    \end{align}
\end{enumerate}
\end{definition}

Note that the items 1) and 2) in the above definition imply a one-to-one correspondence between inputs $X_1,\ldots, X_s$ and outputs $Y_1,\ldots,Y_s$. The item 3) guarantees the security property that no information about any $t$ inputs can be learned from any $s-t$ outputs, termed as the \textit{perfect security}~\cite{ES2021+}. 
\footnote{We remark that the entropy-based perfect security $H(X|Y)=H(X)$ straightforwardly results the statistical distance $SD(P_{X|Y=y},P_X)=0$. In fact, 
$H(X|Y)=H(X) \iff $ $X$ and $Y$ are independent $\iff$ $Pr(x|y) = Pr(x)$ holds for all $x$ and $y$ $\iff$ the statistical distance between $P_{X|Y=y}$ and $P_X$, for all $y$, is $0$.}

\subsection{Combinatorial AONT}

In the meanwhile, the notion of combinatorial all-or-nothing transforms was proposed~\cite{DES2016,Stinson2001}.
We first recall some preliminary definitions. 
An $(N, K, v)$-array $A$ is an $N$ by $K$ array, whose entries are chosen from an alphabet $\Gamma$ of order $v$.  
Let $I\subseteq [K]=\{1,2,\ldots,K\}$, and $A_{I}$ denote the array obtained from $A$ by deleting all the columns indexed by $c\in [K]\setminus I$. 
We say that $A$ is \textit{unbiased} with respect to $I$ if the rows of $A_{I}$ contain every $|I|$-tuple in $\Gamma^{|I|}$ exactly $\frac{N}{v^{|I|}}$ times. 
Based on the unbiased property, the combinatorial AONT is defined as follows. 

\begin{definition}\label{def-combinatorial-AONT}
A \textit{combinatorial $(t,s,v)$-AONT} is a $(v^s, 2s,v)$-array $A$ with columns labeled $1,\ldots, 2s$, which is unbiased with respect to the following subsets of columns: 
\begin{enumerate}
    \item[(1)] $\{1,\ldots, s\}$;
    \item[(2)] $\{s+1,\ldots, 2s\}$;
    \item[(3)] $I\cup J$, for all $I\subseteq \{1,\ldots, s\}$ with $|I|=t$ and all $J\subseteq \{s+1,\ldots, 2s\}$ with $|J|=s-t$. 
\end{enumerate}
\end{definition}

The existence and constructions of combinatorial AONT have been extensively investigated, see~\cite{DES2016,E2021,EGS2018,ES2017,Stinson2001,WCJ,ZZWG} for example. 

To see the connections between the AONT (based on $X_1,\ldots,X_s, Y_1,\ldots,Y_s$) in Definition~\ref{def-AONT} and the combinatorial AONT (based on array $A$) in Definition~\ref{def-combinatorial-AONT}, we can think of the first $s$ columns of $A$ as the inputs $X_1,\ldots,X_s$ and the last $s$ columns of $A$ as the outputs $Y_1,\ldots,Y_s$. 
(In what follows, we also refer to the $i$th column of $A$ as the input $X_i$ and the $(s+i)$th column of $A$ as the output $Y_i$ for $1\le i\le s$ when it is clear from the context.) 
Accordingly, the properties 1) and 2) in Definition~\ref{def-AONT} are equivalent to the properties (1) and (2) in Definition~\ref{def-combinatorial-AONT} in the sense that each of them implies a one-to-one correspondence between inputs and outputs. 
However, in contrast to the perfect security in the property 3) of Definition~\ref{def-AONT}, the item (3) in Definition~\ref{def-combinatorial-AONT} only ensures that knowledge of any $s-t$ outputs does not rule out any possible values for any $t$ inputs, 
which is called the \textit{weak security} in~\cite{ES2021+}. 


It is readily seen that there is a gap between the perfect security and the weak security. Indeed, as pointed out in~\cite{ES2021+}, 
the entropy-based Definition~\ref{def-AONT}   involves the ``security'' of an AONT, while the combinatorial AONT in Definition~\ref{def-combinatorial-AONT} is just a certain mathematical structure and its security properties in general depend on the underlying (prior) probability distribution on the possible inputs. 
Also notice that the input probability distributions together with the combinatorial AONT array induce a probability distribution on the outputs. 
Naturally, the following problem arises. 

\begin{problem}\label{def-question}
What are the security properties of combinatorial AONTs for given (prior) probability distributions on the inputs? 
\end{problem}

 Esfahani and Stinson \cite{ES2021+} provided answers to Problem~\ref{def-question} in the following cases. 
\begin{theorem}\cite[Theorems 2.1 and 2.3]{ES2021+} 
\begin{enumerate}
    \item[(1)] A combinatorial $(t,s,v)$-AONT has weak security provided that all the input $s$-tuples have non-zero probability. 
    \item[(2)] A combinatorial $(t,s,v)$-AONT has perfect security if and only if all the input $s$-tuples are equally probable, i.e., each with probability $1/v^s$. 
\end{enumerate}
\end{theorem}

In addition to the aforementioned, the security properties of combinatorial AONTs are generally unknown.    
In this paper, we aim to explore on the gap between perfect security and weak security, as well as to provide more answers to Problem~\ref{def-question}. Concretely, we consider the typical scenario that all the $s$ inputs take values independently (but not necessarily identically) and quantify $H(\mathcal{X}|\mathcal{Y})$, i.e., the amount of information about any $t$ inputs $\mathcal{X}$ that is not revealed by any $s-t$ outputs $\mathcal{Y}$. In particular, we establish the general lower and upper bounds on $H(\mathcal{X}|\mathcal{Y})$ for combinatorial $(t,s,v)$-AONT (see Theorem~\ref{thm-main-symmetric}) by making use of information-theoretic methods. 
Among others, 
in contrast to the perfect security with $H(\mathcal{X}|\mathcal{Y})=H(\mathcal{X})$ and the weak security with  $0<H(\mathcal{X}|\mathcal{Y})\le H(\mathcal{X})$, we find an interesting phenomenon that for any $t$ inputs $\mathcal{X}$ and any $s-t$ outputs $\mathcal{Y}$, it holds that $$H(\mathcal{X}|\mathcal{Y})\le \min_{\mathcal{X}'}H(\mathcal{X}'),$$ 
where the min is taken over all the $t$ inputs set $\mathcal{X}'\subseteq \{X_1,\ldots,X_s\}$. 
It is also proven that this upper bound can be attained when there are at most $t$ non-uniform inputs (see Theorem~\ref{thm-AONT-at-most-t-uniform}). Some further discussions on the security properties of combinatorial $(t,s,v)$-AONTs in the case when $s$ inputs have partial dependency are also provided.

\subsection{Asymmetric AONT}
On the other hand, very recently, Esfahani and Stinson~\cite{ES2021+asy} generalized the $(t,s,v)$-AONT to the asymmetric $(t_i,t_o,s,v)$-AONT by replacing the parameter $t$ by two parameters $t_i$ and $t_o$ such that $t_i\le t_o$, which has practical applications in the secure distributed storage system~\cite{E2021,KSLC2019} as well.  
Here we formulate the unconditionally secure asymmetric AONT in terms of entropy functions. 
\begin{definition}\label{def-asy-AONT}\rm
Let $X_1,\ldots, X_s$ and $Y_1, \ldots, Y_s$ be input and output random variables respectively, which take values from the finite set $\Gamma$ of size $v$. Let $1\le t_i\le t_o\le s$. 
These $2s$ random variables define an \emph{asymmetric $(t_i,t_o, s, v)$-AONT} provided that the following conditions are satisfied: 
\begin{enumerate}
    \item[1)] $H(Y_1,\ldots,Y_s| X_1,\ldots,X_s) =0$;
    \item[2)] $H(X_1,\ldots,X_s| Y_1,\ldots,Y_s) =0$;
    \item[3)] For all $\mathcal{X} \subseteq \{X_1,\ldots,X_s\}$ with $|\mathcal{X}|=t_i$, and for all $\mathcal{Y} \subseteq \{Y_1,\ldots,Y_s\}$ with $|\mathcal{Y}|=s-t_o$, it holds that 
    \begin{align}\label{eq-perfect-security}
    H(\mathcal{X}|\mathcal{Y}) = H(\mathcal{X}).    
    \end{align}
\end{enumerate}
\end{definition}

The item 3) of the above definition implies the \emph{perfect security} property of an asymmetric $(t_i,t_o, s, v)$-AONT, i.e., knowledge of all but $t_o$ outputs leaves any $t_i$ inputs completely undetermined. 
Note that when $t_i=t_o$, the asymmetric $(t_i,t_o, s, v)$-AONT reduces to the (symmetric) $(t, s, v)$-AONT. Hence the asymmetric $(t_i,t_o, s, v)$-AONT is mainly considered for the case when $t_i<t_o$. 

A formulation of the combinatorial asymmetric AONT is presented in~\cite{ES2021+asy}. 
\begin{definition}\label{def-Asy-combinatorial-AONT}\rm 
A \textit{combinatorial asymmetric $(t_i,t_o,s,v)$-AONT} is a $(v^s, 2s,v)$-array $A$ with columns labeled $1,\ldots, 2s$, which is unbiased with respect to the following subsets of columns: 
\begin{enumerate}
    \item[(1)] $\{1,\ldots, s\}$;
    \item[(2)] $\{s+1,\ldots, 2s\}$;
    \item[(3)] $I\cup J$, for all $I\subseteq \{1,\ldots, s\}$ with $|I|=t_i$ and all $J\subseteq \{s+1,\ldots, 2s\}$ with $|J|=s-t_o$. 
\end{enumerate}
\end{definition} 
By Definitions~\ref{def-combinatorial-AONT} and \ref{def-Asy-combinatorial-AONT}, it is readily seen that a combinatorial $(t,s,v)$-AONT is a combinatorial asymmetric $(t_i,t,s,v)$-AONT for any $1\le t_i\le t$. 
The \emph{weak security} property of a combinatorial asymmetric $(t_i,t_o,s,v)$-AONT specifies that knowledge of any $s-t_o$ outputs does not rule out any possible values for any $t_i$ inputs. 
Notice that when $t_i<t_o$, the weak security property of a combinatorial asymmetric $(t_i,t,s,v)$-AONT does not necessarily require that all the input $s$-tuples are with positive probability. 
Inspired by this, Esfahani and Stinson~\cite{ES2021+asy} relaxed the requirements of unbiased property on combinatorial asymmetric AONT to the covering property, and then introduced the notion of combinatorial asymmetric weak-AONT. 

Let $A$ be an $(N, K, v)$-array over the alphabet $\Gamma$ of order $v$.   
Let $I\subseteq [K]=\{1,2,\ldots,K\}$ and $A_{I}$ be the array obtained from $A$ by deleting all the columns indexed by $c\in [K]\setminus I$. 
We say that $A$ is \textit{covering} with respect to $I$ if the rows of $A_{I}$ contain every $|I|$-tuple in $\Gamma^{|I|}$ at least once.  
\begin{definition}\label{def-Asy-combinatorial-weak-AONT}\rm 
A \textit{combinatorial asymmetric $(t_i,t_o,s,v)$-weak-AONT} is a $(v^s, 2s,$ $v)$-array $A$ with columns labeled $1,\ldots, 2s$, which is covering with respect to the following subsets of columns: 
\begin{enumerate}
    \item[(1)] $\{1,\ldots, s\}$;
    \item[(2)] $\{s+1,\ldots, 2s\}$;
    \item[(3)] $I\cup J$, for all $I\subseteq \{1,\ldots, s\}$ with $|I|=t_i$ and all $J\subseteq \{s+1,\ldots, 2s\}$ with $|J|=s-t_o$. 
\end{enumerate}
\end{definition} 

It is easily seen that a combinatorial asymmetric $(t_i,t_o,s,v)$-AONT is a combinatorial asymmetric $(t_i,t_o,s,v)$-weak-AONT, but not vice versa. 
The existence and constructions of combinatorial asymmetric (weak)-AONT have been studied in~\cite{E2021,ES2021+asy}. 
Notice that the combinatorial asymmetric AONT and weak-AONT are mathematical structures. In terms of their security, the following problem appears. 

\begin{problem}\label{def-question-asy}
What are the security properties of combinatorial asymmetric $(t_i,t_o,s,v)$-(weak)-AONTs for given (prior) probability distributions on the inputs? 
\end{problem}

The following answers to Problem~\ref{def-question-asy} can be found in~\cite{ES2021+asy}. 
\begin{theorem}\cite[Theorem 2.3]{ES2021+asy} 
\begin{enumerate}
    \item[(1)] A combinatorial asymmetric $(t_i,t_o,s,v)$-(weak)-AONT has weak security if all the input $s$-tuples have positive probability. 
    \item[(2)] A combinatorial asymmetric $(t_i,t_o,s,v)$-AONT has perfect security if every input $s$-tuple occurs with the same probability $1/v^s$. 
\end{enumerate}
\end{theorem}

Except for the above-mentioned, the security properties of combinatorial asymmetric AONT 
are unknown in general. 
This paper is devoted to exploring their security properties which are sandwiched between the known perfect security and weak security, as well as to providing answers to Problem~\ref{def-question-asy}. 
Again, we consider the typical scenario that all the $s$ inputs take values independently but not necessarily identically, and quantify the amount of information $H(\mathcal{X}|\mathcal{Y})$ about any $t$ inputs $\mathcal{X}$ that is not learned by any $s-t$ outputs $\mathcal{Y}$. 
By generalizing the discussions on combinatorial AONT, we establish general lower and upper bounds on $H(\mathcal{X}|\mathcal{Y})$ for combinatorial asymmetric AONTs 
(see Theorem~\ref{thm-main-Asymmetric}).
It is also shown that the established bounds could be attained in certain cases. 
In addition, some discussions on the differences of security properties between combinatorial asymmetric AONTs and combinatorial (symmetric) AONTs are presented as well.


The remainder of this paper is organized as follows. 
Section~\ref{sec-sym-AONT} establishes general lower and upper bounds for combinatorial AONT, and shows that the derived bounds could be achieved in certain cases. 
Section~\ref{sec-Asym-AONT} 
and Section~\ref{sec-Asym-weak-AONT} 
discuss the security properties for combinatorial asymmetric AONT and combinatorial asymmetric weak-AONT respectively. 
Finally Section~\ref{sec-conclusion} concludes this paper.

\section{AONT with independent inputs}
\label{sec-sym-AONT}

In this section, we first prove the general lower and upper bounds on $H(\mathcal{X}|\mathcal{Y})$ for  combinatorial AONT with independent inputs. 
Then we show that the derived bounds can be achieved in certain cases.

\subsection{General bounds for combinatorial AONT}

In this subsection, we establish the following theorem. 

\begin{theorem}\label{thm-main-symmetric} 
Let array $A\in \Gamma^{v^s\times 2s}$ be a combinatorial $(t,s,v)$-AONT whose columns are with respect to random variables $X_1,\ldots,X_s,Y_1,\ldots,Y_s$ respectively. Let $P_1,\ldots,P_s$ be the corresponding probability distributions of $X_1,\ldots,X_s$, which are mutually independent~\footnote{For a collection of random variables, the random variables are called mutually independent if each random variable is independent of any combination of other random variables in the collection.}~and take values from $\Gamma$. Then for any input sets $\mathcal{X},\mathcal{X}'\subseteq \{X_1,\ldots,$ $X_s\}$ such that $|\mathcal{X}|=|\mathcal{X}'|=t$, and any $s-t$ outputs $\mathcal{Y}\subseteq \{Y_1,\ldots,Y_s\}$ such that $|\mathcal{Y}|=s-t$, the followings hold.  
\begin{enumerate}
    \item[(1)]      
    \begin{align}\label{eq-thm-main-3}
        H(\mathcal{X}|\mathcal{Y}) = H(\mathcal{X}'|\mathcal{Y}). 
    \end{align}
    \item[(2)] 
    \begin{align}\label{eq-thm-main-lower}
    H(\mathcal{X}|\mathcal{Y}) \ge  \max\bigg\{0,\sum\limits_{i\in [s]} H(X_i)-(s-t)\log(v)\bigg\}.  
    \end{align}
    \item[(3)] 
    \begin{align}\label{eq-thm-main-upper}
    H(\mathcal{X}|\mathcal{Y}) 
        \le  
        \min\limits_{\mathcal{X}'\subseteq \{X_1,\ldots,X_s\}
        ,\atop |\mathcal{X}'|=t} H(\mathcal{X}')
        =
        \min\limits_{I\subseteq [s],\atop|I|=t}\, \sum\limits_{i\in I} H(X_i)\le  H(\mathcal{X}).
    \end{align}
\end{enumerate}
\end{theorem}

In order to prove Theorem~\ref{thm-main-symmetric}, we first show the following lemma. 

\begin{lemma}\label{lemma-symmetric-general}
Under the assumption of Theorem~\ref{thm-main-symmetric}, for any $t$ inputs $\mathcal{X}\subseteq \{X_1,\ldots,$ $X_s\}$ such that $|\mathcal{X}|=t$ and any $s-t$ outputs $\mathcal{Y}\subseteq \{Y_1,\ldots,Y_s\}$ such that $|\mathcal{Y}|=s-t$, we have 
\begin{align}
     H(\mathcal{X}|\mathcal{Y}) =  \sum_{i\in [s]} H(X_i) - H(\mathcal{Y}).
\end{align}
\end{lemma}
\begin{proof} 
Let $\overline{\mathcal{X}} = \{X_1,\ldots,X_s\}\setminus \mathcal{X}$. 
Then by the definition, we have 
\begin{align}
    H(\mathcal{X}|\mathcal{Y}) 
    \nonumber 
    & = H(\mathcal{X},\mathcal{Y})-H(\mathcal{Y})\\
    \nonumber 
    &= \sum_{\mathbf{u}\in \Gamma^t,  \mathbf{v}\in \Gamma^{s-t}} Pr[\mathcal{X}=\mathbf{u},\mathcal{Y}=\mathbf{v}] \log \frac{1}{Pr[\mathcal{X}=\mathbf{u},\mathcal{Y}=\mathbf{v}]}-H(\mathcal{Y})\\
    \label{eq-y-to-x-bar}
    &= \sum_{\mathbf{u}\in \Gamma^t,  \mathbf{u}'\in \Gamma^{s-t}} Pr[\mathcal{X}=\mathbf{u},\overline{\mathcal{X}}=\mathbf{u}'] \log \frac{1}{Pr[\mathcal{X}=\mathbf{u},\overline{\mathcal{X}}=\mathbf{u}']}
    -H(\mathcal{Y})\\
    \nonumber
    & = H(\mathcal{X},\overline{\mathcal{X}}) - H(\mathcal{Y})\\
    \label{eq-final-H-to-sum}
    & = \sum_{i\in [s]} H(X_i) - H(\mathcal{Y})
\end{align}
where~\eqref{eq-y-to-x-bar} follows from the unbiased property of the combinatorial $(t,s,v)$-AONT, i.e.,  each pair of $\mathcal{X}=\mathbf{u},\mathcal{Y}=\mathbf{v}$ uniquely determines a $\mathbf{u}'\in \Gamma^{s-t}$ such that $\overline{\mathcal{X}}=\mathbf{u}'$; and \eqref{eq-final-H-to-sum} follows from the assumption that all $X_i$ are mutually independent. This proves the lemma. 
\end{proof}

According to Lemma~\ref{lemma-symmetric-general}, estimating $H(\mathcal{X}|\mathcal{Y})$ requires the analysis on $H(\mathcal{Y})$ for given probability distributions on inputs $X_1,\ldots,X_s$.

To prove Theorem~\ref{thm-main-symmetric}, we will also make use of the following lemma from the Jensen's inequality on the convex function $f(x)=x\log(x)$. 

\begin{lemma}[\cite{CT-book,SP}]
\label{lemma-Jensen-ineq}  
For any $\{m_i\}_{i=1}^n$ with $m_i\ge 0$ and $\sum_{i=1}^n m_i =1$ and non-negative sequence $\{x_i\}_{i=1}^n$, we have 
\begin{align}\label{eq-Jensen-ineq}
    \Bigg(\sum_{i=1}^n m_ix_i\Bigg) \log \Bigg(\sum_{i=1}^n m_ix_i\Bigg) 
    \le \sum_{i=1}^n m_ix_i\log(x_i). 
\end{align}
Furthermore, when $m_i>0$ for all $i\in [n]$, the equality of \eqref{eq-Jensen-ineq} holds if and only if $x_1=\cdots=x_n$. 
\end{lemma}

Now we are ready to prove Theorem~\ref{thm-main-symmetric}. 
\begin{proof}[Proof of Theorem~\ref{thm-main-symmetric}] 

(1)
The equality \eqref{eq-thm-main-3} follows immediately from Lemma~\ref{lemma-symmetric-general}, i.e., 
\begin{align*}
     H(\mathcal{X}|\mathcal{Y}) = H(\mathcal{X}'|\mathcal{Y}) = \sum_{i\in [s]} H(X_i) - H(\mathcal{Y}).
\end{align*}

(2) 
Since the entropy $H(\mathcal{X}|\mathcal{Y})$ is non-negative and according to Lemma~\ref{lemma-symmetric-general}, 
\begin{align*}
     H(\mathcal{X}|\mathcal{Y})
     &= \sum_{i\in [s]} H(X_i) - H(\mathcal{Y})\\
     &\ge \sum_{i\in [s]} H(X_i) - \sum_{Y\in \mathcal{Y}}H(Y) \\
     &\ge  \sum_{i\in [s]} H(X_i) - (s-t)\log(v)
\end{align*}
where the first inequality follows from the relation $H(Y,Y')\le H(Y)+H(Y')$; and the second inequality follows since $H(Y)\le \log(v)$ for any output $Y$. 
Thus \eqref{eq-thm-main-lower} follows. 

(3)
Recall the assumption that $X_1,\ldots,X_s$ are mutually independent.
Let $\mathcal{X}_{\max}\subseteq \{X_1,\ldots,X_s\}$ such that $|\mathcal{X}_{\max}|=s-t$ denote the  collection of input random variables according to the largest $s-t$ entropy values among $H(X_1),\ldots,$ $H(X_s)$.  
Denote $\overline{\mathcal{X}_{\max}} = \{X_1,\ldots,X_s\} \setminus \mathcal{X}_{\max}$. Clearly, $|\overline{\mathcal{X}_{\max}}|=t$. We now claim that for any $s-t$ outputs $\mathcal{Y}\subseteq \{Y_1,\ldots,Y_s\}$ such that $|\mathcal{Y}|=s-t$, we have 
\begin{align}\label{eq--thm5.2-proof-claim-rhs}
    H(\mathcal{Y})\ge H(\mathcal{X}_{\max}). 
\end{align}
In fact, by the assumption of combinatorial $(t,s,v)$-AONT, we have 
\begin{align}
    H(\mathcal{Y}) 
    \nonumber 
    & = -\sum_{\mathbf{v}\in \Gamma^{s-t}} Pr[\mathcal{Y}=\mathbf{v}] \log  \big(Pr[\mathcal{Y}=\mathbf{v}]\big)\\
    \label{eq-thm5.2-proof-use-CANOT}
    & = -\sum_{\mathbf{v}\in \Gamma^{s-t}} 
       \Bigg( \sum_{\mathbf{u}'\in \Gamma^t} Pr[\mathcal{X}_{\max}=\mathbf{u},\overline{\mathcal{X}_{\max}}=\mathbf{u}'] \Bigg)
       \cdot
       \log 
       \Bigg( \sum_{\mathbf{u}'\in \Gamma^t} Pr[\mathcal{X}_{\max}=\mathbf{u},\overline{\mathcal{X}_{\max}}=\mathbf{u}'] \Bigg) \\
    \label{eq-thm5.2-proof-use-independent}
    & = -\sum_{\mathbf{v}\in \Gamma^{s-t}} 
       \Bigg( \sum_{\mathbf{u}'\in \Gamma^t} Pr[\overline{\mathcal{X}_{\max}}=\mathbf{u}']\cdot Pr[\mathcal{X}_{\max}=\mathbf{u}]\Bigg)
       \cdot
       \log 
       \Bigg( \sum_{\mathbf{u}'\in \Gamma^t} Pr[\overline{\mathcal{X}_{\max}}=\mathbf{u}']\cdot Pr[\mathcal{X}_{\max}=\mathbf{u}]\Bigg) \\
    \label{eq-thm5.2-proof-use-Jensen}
    & \ge 
       -\sum_{\mathbf{v}\in \Gamma^{s-t}}  
       \sum_{\mathbf{u}'\in \Gamma^t} Pr[\overline{\mathcal{X}_{\max}}=\mathbf{u}']\cdot Pr[\mathcal{X}_{\max}=\mathbf{u}] \cdot 
       \log\big( Pr[\mathcal{X}_{\max}=\mathbf{u}] \big)\\ 
    \label{eq-thm5.2-proof-switch-sums}
    & = 
       -
       \sum_{\mathbf{u}'\in \Gamma^t} Pr[\overline{\mathcal{X}_{\max}}=\mathbf{u}']\cdot
       \sum_{\mathbf{v}\in \Gamma^{s-t}}  
       Pr[\mathcal{X}_{\max}=\mathbf{u}] \cdot 
       \log\big( Pr[\mathcal{X}_{\max}=\mathbf{u}] \big)\\    
    \label{eq-thm5.2-proof-v-to-u}
    & = 
       -
       \sum_{\mathbf{u}'\in \Gamma^t} Pr[\overline{\mathcal{X}_{\max}}=\mathbf{u}']\cdot
       \sum_{\mathbf{u}\in \Gamma^{s-t}}  
       Pr[\mathcal{X}_{\max}=\mathbf{u}] \cdot 
       \log\big( Pr[\mathcal{X}_{\max}=\mathbf{u}] \big)\\ 
    \nonumber 
    & =
       \sum_{\mathbf{u}'\in \Gamma^t} Pr[\overline{\mathcal{X}_{\max}}=\mathbf{u}']\cdot 
       H(\mathcal{X}_{\max})\\
    \nonumber 
    & = H(\mathcal{X}_{\max}) 
\end{align}
where \eqref{eq-thm5.2-proof-use-CANOT} follows from the unbiased property (on $\overline{\mathcal{X}_{\max}}$ and $\mathcal{Y}$) of the combinatorial $(t,s,v)$-AONT and the $(s-t)$-tuple $\mathbf{u}$ such that  $\mathcal{X}_{\max}=\mathbf{u}$ is uniquely determined by $\mathcal{Y}=\mathbf{v}$ and $\overline{\mathcal{X}_{\max}}=\mathbf{u}'$; 
\eqref{eq-thm5.2-proof-use-independent} follows from the mutual independence among $X_1,\ldots,X_s$; 
\eqref{eq-thm5.2-proof-use-Jensen} follows from Lemma~\ref{lemma-Jensen-ineq} with $\sum_{\mathbf{u}'\in \Gamma^t} Pr[\overline{\mathcal{X}_{\max}}=\mathbf{u}'] =1$ and $Pr[\mathcal{X}_{\max}=\mathbf{u}]\ge 0$; 
\eqref{eq-thm5.2-proof-switch-sums} follows by switching the two sums on $\mathbf{v}$ and $\mathbf{u}'$ which is doable due to their independence; 
and \eqref{eq-thm5.2-proof-v-to-u} follows from the assumption of combinatorial $(t,s,v)$-AONT and its unbaised property. Thus the claim holds. 

Combining Lemma~\ref{lemma-symmetric-general} with \eqref{eq--thm5.2-proof-claim-rhs}, we obtain 
\begin{align*}
     H(\mathcal{X}|\mathcal{Y})
     &= \sum_{i\in [s]} H(X_i) - H(\mathcal{Y})\\
     &\le \sum_{i\in [s]} H(X_i) - H(\mathcal{X}_{\max})  \\
     & = \sum_{i\in [s]} H(X_i) - \sum_{X\in \mathcal{X}_{\max}}H(X) \\
     &= \min_{I\subseteq [s],\atop|I|=t}\, \sum_{i\in I} H(X_i). 
\end{align*}
This completes the proof. 
\end{proof}

Regarding the item (1) of  Theorem~\ref{thm-main-symmetric}, we remark that for any distinct output sets  $\mathcal{Y},\mathcal{Y}'\subseteq \{Y_1,\ldots,Y_s\}$ such that $\mathcal{Y}\ne \mathcal{Y}'$ and $|\mathcal{Y}|=|\mathcal{Y}'|=s-t$, the relation $H(\mathcal{X}|\mathcal{Y}) = H(\mathcal{X}'|\mathcal{Y}')$ does not hold in general, even for the case when  $\mathcal{X}=\mathcal{X}'$ (see Example~\ref{example-1}). 
However the equality $H(\mathcal{X}|\mathcal{Y}) = H(\mathcal{X}'|\mathcal{Y}')$ always holds in the case when at most $t$ inputs are with non-uniform distributions (see Section~\ref{subsec-AONT-at-most-t-uniform}).

Also we would remark that 
the upper bound in Theorem~\ref{thm-main-symmetric} are tight in the sense that it can be achieved in certain cases (see Section~\ref{subsec-AONT-at-most-t-uniform}).  
In contrast, it is also worth noting that in general  $H(\mathcal{X}|\mathcal{Y})$ depends on the input distributions and the upper bound in Theorem~\ref{thm-main-symmetric} might not always be attained, see Example~\ref{example-1} below considering the case when more than $t$ inputs are with non-uniform distributions. 


\begin{example}\label{example-1}\rm 
Consider the following combinatorial $(1,2,3)$-AONT over the alphabet $\Gamma=\{a, b, c\}$ shown in Table~\ref{table-1} as in \cite{ES2021+}. 

\begin{table}[h]
 \centering
  \begin{tabular}{|c|c||c|c|}
   \hline
   $X_1$ & $X_2$ & $Y_1$ & $Y_2$\\
   \hline
   $a$ & $a$ & $a$ & $a$\\
   $a$ & $b$ & $c$ & $b$\\
   $a$ & $c$ & $b$ & $c$\\
   $b$ & $a$ & $b$ & $b$\\
   $b$ & $b$ & $a$ & $c$\\
   $b$ & $c$ & $c$ & $a$\\
   $c$ & $a$ & $c$ & $c$\\
   $c$ & $b$ & $b$ & $a$\\
   $c$ & $c$ & $a$ & $b$\\
   \hline
  \end{tabular}
  \vskip 0.3cm
  \caption{A combinatorial $(1,2,3)$-AONT over the alphabet $\{a, b, c\}$}
  \label{table-1}
\end{table}
Suppose that
\begin{align*}
    Pr[X_1=a]=\frac{1}{4},\  Pr[X_1=b]=\frac{1}{8},\  Pr[X_1=c]=\frac{5}{8},\\ Pr[X_2=a]=\frac{1}{3},\  Pr[X_2=b]=\frac{1}{6},\  Pr[X_2=c]=\frac{1}{2}.
\end{align*}
Plugging into Table~\ref{table-1} gives 
\begin{align*}
    Pr[Y_1=a]=\frac{5}{12},\  Pr[Y_1=b]=\frac{13}{48},\  Pr[Y_1=c]=\frac{5}{16},\\ Pr[Y_2=a]=\frac{1}{4},\  Pr[Y_2=b]=\frac{19}{48},\  Pr[Y_2=c]=\frac{17}{48}.
\end{align*}
Then by the definition, it is easy to calculate 
\begin{align*}
    H(X_1)=1.298795, \ H(X_2)=1.459148,\ H(Y_1)=1.561053, \ H(Y_2)=1.559607.
\end{align*}
According to Lemma~\ref{lemma-symmetric-general} and \eqref{eq-thm-main-3}, we obtain 
\begin{align*}
    H(X_1|Y_1) = H(X_2|Y_1) 
    &= H(X_1) + H(X_2) - H(Y_1) 
    = 1.196889<\min\{H(X_1),H(X_2)\}, \\
    H(X_1|Y_2) = H(X_2|Y_2) 
    &= H(X_1) + H(X_2) - H(Y_2) 
    = 1.198335<\min\{H(X_1),H(X_2)\}, 
\end{align*}
while it is easy to see $H(X_1|Y_1) \neq H(X_1|Y_2)$ and
$H(X_2|Y_1) \neq H(X_2|Y_2)$. 

\end{example}

\subsection{AONT with at most $t$ non-uniform inputs}
\label{subsec-AONT-at-most-t-uniform} 

In this subsection, we consider the situation where at most $t$ inputs are with non-uniform distributions. It is shown that the bounds in Theorem~\ref{thm-main-symmetric} turn to be tight in this case.   

\begin{theorem}\label{thm-AONT-at-most-t-uniform}
Under the assumption of Theorem~\ref{thm-main-symmetric}, if at most $t$ of $P_1,\ldots,P_s$ are non-uniform, then for any $t$ inputs $\mathcal{X}\subseteq \{X_1,\ldots,X_s\}$ such that $|\mathcal{X}|=t$ and any $s-t$ outputs $\mathcal{Y}\subseteq \{Y_1,\ldots,Y_s\}$ such that $|\mathcal{Y}|=s-t$, we have 
\begin{align}\label{eq-thm-main-1}
        H(\mathcal{X}|\mathcal{Y}) = \min_{\mathcal{X}'\subseteq \{X_1,\ldots,X_s\},\atop |\mathcal{X}'|=t} H(\mathcal{X}')= \min_{I\subseteq [s],\atop|I|=t}\, \sum_{i\in I} H(X_i) \le  H(\mathcal{X}).
    \end{align}
    In other words, if $P_{i_1},\ldots,P_{i_r}$, where $1\le r\le t$, are non-uniform, and all the others in $P_1,\ldots,P_s$ are uniform, then 
    \begin{align}\label{eq-thm-main-2}
        H(\mathcal{X}|\mathcal{Y}) = H(X_{i_1}) + \cdots + H(X_{i_r}) + (t-r)\log(v).
    \end{align}
    
    
\end{theorem}

\begin{proof}
The conclusion follows from the lower and upper bounds in Theorem~\ref{thm-main-symmetric} in which all the inputs in $\mathcal{X}_{\max}$ are with the uniform distribution.  
\end{proof}

Notice that the above theorem assumes the mutually independence among all inputs $X_1,\ldots,X_s$. 
In the following, we show that the conclusion as in Theorem~\ref{thm-AONT-at-most-t-uniform} also holds even for the case that a local dependence among at most $t$ inputs exists.  Precisely, 
we have the following theorem. 

\begin{theorem}\label{thm-AONT-at-most-t-uniform-depend}
Let array $A\in \Gamma^{v^s\times 2s}$ be a combinatorial $(t,s,v)$-AONT whose columns are with respect to random variables $X_1,\ldots,X_s,Y_1,\ldots,Y_s$ respectively. 
If there are at most $t$ inputs $\mathcal{X}_0\subseteq \{X_1,\ldots,X_s\}$ such that $|\mathcal{X}_0|\le t$ are dependent, and all the other inputs in $\overline{\mathcal{X}}_0 = \{X_1,\ldots,X_s\}\setminus \mathcal{X}_0$ are mutually independent~\footnote{It also implicitly assumes that any input in $\overline{\mathcal{X}_0}$ is independent of any combination of other inputs in $\{X_1, \ldots, X_s\}$.} and with the uniform distribution, then for any $t$ inputs $\mathcal{X}\subseteq \{X_1,\ldots,X_s\}$ such that $|\mathcal{X}|=t$ and any $s-t$ outputs $\mathcal{Y}\subseteq \{Y_1,\ldots,Y_s\}$ such that $|\mathcal{Y}|=s-t$, we have 
\begin{align}\label{eq-thm-main-1}
        H(\mathcal{X}|\mathcal{Y}) = \min_{\mathcal{X}'\subseteq \{X_1,\ldots,X_s\},\atop |\mathcal{X}'|=t} H(\mathcal{X}').
\end{align}
\end{theorem}

\begin{proof} 
First, similar to Lemma~\ref{lemma-symmetric-general}, 
for any $t$ inputs $\mathcal{X}\subseteq \{X_1,\ldots,X_s\}$ such that $|\mathcal{X}|=t$ and any $s-t$ outputs $\mathcal{Y}\subseteq \{Y_1,\ldots,Y_s\}$ such that $|\mathcal{Y}|=s-t$, we have 
\begin{align}\label{eq-lemma-depend}
     H(\mathcal{X}|\mathcal{Y}) =  H(X_1,\ldots,X_s) - H(\mathcal{Y}).
\end{align}

According to the assumption 
we can assume that $\mathcal{X}_U\subseteq \{X_1,\ldots,X_s\}$ such that $|\mathcal{X}_U|=s-t$ be a set of $s-t$ 
mutually independent inputs with uniform probability distribution. 
Let $\overline{\mathcal{X}_U} = \{X_1,\ldots,X_s\}\setminus \mathcal{X}_U$. Clearly $|\overline{\mathcal{X}_U}|=t$ and $\mathcal{X}_0 \subseteq \overline{\mathcal{X}_U}$.     

We first claim that in this setting for any $s-t$ outputs $\mathcal{Y}\subseteq \{Y_1,\ldots,Y_s\}$ such that $|\mathcal{Y}|=s-t$,  
\begin{align}\label{eq-claim-Hy}
    H(\mathcal{Y}) = H(\mathcal{X}_U) = (s-t)\log(v). 
\end{align}
Indeed, the above \eqref{eq-claim-Hy} can be verified by following the same line as the argument for \eqref{eq--thm5.2-proof-claim-rhs}, in which the equality in \eqref{eq-thm5.2-proof-use-Jensen} is achieved according to Lemma~\ref{lemma-Jensen-ineq} with   $Pr[\mathcal{X}_{\max}=\mathbf{u}]=Pr[\mathcal{X}_U=\mathbf{u}]=\frac{1}{v^{s-t}}$ for any $\mathbf{u}\in \Gamma^{s-t}$. Here we would provide an alternative (and simpler) proof for \eqref{eq-claim-Hy} as below by directly showing that the probability distribution on $\mathcal{Y}$ is uniform. 

It suffices to prove that $Pr[\mathcal{Y}=\mathbf{v}] = \frac{1}{v^{s-t}}$ for any $\mathbf{v}\in \Gamma^{s-t}$, i.e., outputs $\mathcal{Y}$ take values on all $(s-t)$-tuples with equal probability. In fact, 
\begin{align}
    Pr[\mathcal{Y}=\mathbf{v}] 
    \label{eq-use-unbiased-y-xU-bar}
    & = \sum_{\mathbf{u}'\in \Gamma^t} Pr[\mathcal{X}_U=\mathbf{u},\overline{\mathcal{X}_U}=\mathbf{u}']\\
    \label{eq-use-mutually-independent}
    &=  \sum_{\mathbf{u}'\in \Gamma^t} Pr[\mathcal{X}_U=\mathbf{u}]\cdot  Pr[\overline{\mathcal{X}_U}=\mathbf{u}']\\
    \label{eq-use-uniform-U}
    & = \frac{1}{v^{s-t}}\cdot  \sum_{\mathbf{u}'\in \Gamma^t} Pr[\overline{\mathcal{X}_U}=\mathbf{u}']\\
    \nonumber
    & = \frac{1}{v^{s-t}}
\end{align}
where~\eqref{eq-use-unbiased-y-xU-bar} follows from the unbiased property (on $\overline{\mathcal{X}_U}$ and $\mathcal{Y}$) of the combinatorial $(t,s,v)$-AONT and hence the $(s-t)$-tuple $\mathbf{u}$ with $\mathcal{X}_U=\mathbf{u}$ is uniquely determined by $\overline{\mathcal{X}_U}=\mathbf{u}'$ and $\mathcal{Y}=\mathbf{v}$; \eqref{eq-use-mutually-independent} follows since 
$\mathcal{X}_U$ and $\overline{\mathcal{X}_U}$
are independent; and \eqref{eq-use-uniform-U} follows from the assumption that all inputs in $\mathcal{X}_U$ have uniform distribution.  Thus the claim follows.  

Combining \eqref{eq-lemma-depend} with \eqref{eq-claim-Hy}, we obtain 
\begin{align}
     H(\mathcal{X}|\mathcal{Y})  
     = H(X_1,\ldots,X_s) -  (s-t)\log(v) 
     = \min_{\mathcal{X}'\subseteq \{X_1,\ldots,X_s\},\atop |\mathcal{X}'|=t} H(\mathcal{X}'), 
\end{align}
as required. 
\end{proof}

A direct consequence of Theorem~\ref{thm-AONT-at-most-t-uniform-depend} is as follows. 
\begin{corollary}\label{coro-depend}
Under the assumption of Theorem~\ref{thm-AONT-at-most-t-uniform-depend}, 
if $|\mathcal{X}_0|=t$, 
then 
\begin{align*}
     H(\mathcal{X}|\mathcal{Y}) = H(\mathcal{X}_0). 
\end{align*}
\end{corollary}

\begin{remark}
We would remark that the scenario in Corollary~\ref{coro-depend} was investigated in the context of ``randomized AONTs'' in \cite{ES2021+}. It is also worth noting that the conclusion  Corollary~\ref{coro-depend} here 
generalizes Theorem 3.1 in \cite{ES2021+} in the sense that only a special case $H(\mathcal{X}_0|\mathcal{Y}) = H(\mathcal{X}_0)$ is verified in \cite[Theorem 3.1]{ES2021+}. 
\end{remark}

\begin{remark}
Notice that all the inputs $s$-tuples are equally probable is equivalent to  all the $s$ inputs are mutually independent and with uniform distribution. 
As a direct consequence of Theorem~\ref{thm-AONT-at-most-t-uniform}, for a combinatorial $(t,s,v)$-AONT where all the $s$ inputs are equally probable, we have the perfect security property that $H(\mathcal{X}|\mathcal{Y})=H(\mathcal{X})=t\log(v)$ for any $t$ inputs $\mathcal{X}$ and any $s-t$ outputs $\mathcal{Y}$. This coincides with the conclusion in \cite[Theorem 2.3]{ES2021+}, where the necessary condition for the perfect security of combinatorial AONTs is also discussed. 
In addition,
Theorem~\ref{thm-AONT-at-most-t-uniform} works further for the case when the inputs $s$-tuples are not equally probable (see Example~\ref{example-aont-beyond-equi}). 
\end{remark}

\begin{example}\label{example-aont-beyond-equi}
Consider the combinatorial $(1,2,3)$-AONT as in Table~\ref{table-1}.
Suppose that
\begin{align*}
    Pr[X_1=a]=\frac{1}{3},\  Pr[X_1=b]=\frac{1}{3},\  Pr[X_1=c]=\frac{1}{3},\\ Pr[X_2=a]=\frac{1}{3},\  Pr[X_2=b]=\frac{1}{6},\  Pr[X_2=c]=\frac{1}{2}.
\end{align*}
Plugging into Table~\ref{table-1} gives 
\begin{align*}
    Pr[Y_1=a]=\frac{1}{3},\  Pr[Y_1=b]=\frac{1}{3},\  Pr[Y_1=c]=\frac{1}{3},\\ Pr[Y_2=a]=\frac{1}{3},\  Pr[Y_2=b]=\frac{1}{3},\  Pr[Y_2=c]=\frac{1}{3}.
\end{align*}
Then by the definition, it is easy to calculate 
\begin{align*}
    H(X_1)=1.584963, \ H(X_2)=1.459148,\ H(Y_1)=1.584963, \ H(Y_2)=1.584963.
\end{align*}
According to Lemma~\ref{lemma-symmetric-general} and \eqref{eq-thm-main-3}, we obtain 
\begin{align*}
    H(X_1|Y_1) &= H(X_2|Y_1)= H(X_1) + H(X_2) - H(Y_1)
    = 1.459148=\min\{H(X_1),H(X_2)\}, \\
    H(X_1|Y_2) &= H(X_2|Y_2)= H(X_1) + H(X_2) - H(Y_2) 
    = 1.459148=\min\{H(X_1),H(X_2)\}. 
\end{align*}
\end{example}

\section{Asymmetric AONT with independent inputs}
\label{sec-Asym-AONT}

In this section we investigate the security properties of combinatorial asymmetric AONT with independent inputs. 
In particular, we establish the following theorem.

\begin{theorem}\label{thm-main-Asymmetric} 
Let array $A\in \Gamma^{v^s\times 2s}$ be a combinatorial asymmetric $(t_i,t_o,s,v)$-AONT whose columns are with respect to random variables $X_1,\ldots,X_s,Y_1,\ldots,Y_s$ respectively. Let $P_1,\ldots,P_s$ be the corresponding probability distributions of $X_1,$ $\ldots,X_s$, which are  
mutually independent and take values from $\Gamma$. 
Then for any $t_i$ inputs $\mathcal{X}\subseteq \{X_1,\ldots,X_s\}$ such that $|\mathcal{X}|=t_i$ and any $s-t_o$ outputs $\mathcal{Y}\subseteq \{Y_1,\ldots,Y_s\}$ such that $|\mathcal{Y}|=s-t_o$, the followings hold.  
\begin{enumerate}
    \item[(1)] \begin{align}\label{eq-thm-Asy-main-1-ld}
        H(\mathcal{X}|\mathcal{Y})  \ge \max\bigg\{0,\sum_{i=1}^s H(X_i) - (s-t_i)\log(v) \bigg\}. 
    \end{align}
    \item[(2)] 
    \begin{align}\label{eq-thm-Asy-main-2}
        H(\mathcal{X}|\mathcal{Y}) 
        \le \min\bigg\{H(\mathcal{X}) , 
        & \min_{I\subseteq [s],\atop|I|=t_i}\, \sum_{i\in I} H(X_i) + (t_o-t_i)\log(v),\\
        &\min_{I\subseteq [s],\atop|I|=t_i}\, \sum_{i\in I} H(X_i) + s\cdot\log(v)-\sum^s_{i=1}H(X_i)
        \bigg\}.
    \end{align}
    
\end{enumerate}
\end{theorem}

We would remark that Theorem~\ref{thm-main-Asymmetric} for combinatorial  asymmetric $(t_i,t_o,s,v)$-AONTs can be seen as a generalization of Theorem~\ref{thm-main-symmetric} for combinatorial $(t,s,v)$-AONTs in the sense that 
the upper and lower bounds in Theorem~\ref{thm-main-Asymmetric} could deduce the bounds in Theorem~\ref{thm-main-symmetric} by letting $t_i=t_o=t$. 
However, when $t_i<t_o$, the property (1) of Theorem~\ref{thm-main-symmetric} does not hold in general for combinatorial  asymmetric $(t_i,t_o,s,v)$-AONTs (see Example~\ref{example-asy}). 
Furthermore, 
in contrast to Theorem~\ref{thm-main-symmetric}, the quantification of $H(\mathcal{X}|\mathcal{Y})$ for combinatorial  asymmetric $(t_i,t_o,s,v)$-AONTs  cannot be upper bounded by $\min\limits_{\mathcal{X}'}H(\mathcal{X}')$ in general (see Example~\ref{example-asy}). 

\subsection{A general lemma}
In order to prove Theorem~\ref{thm-main-Asymmetric}, we first prove the following lemma. 

\begin{lemma}\label{lemma-Asymmetric-general}
Under the assumption of Theorem~\ref{thm-main-Asymmetric}, for any $t_i$ inputs $\mathcal{X}\subseteq \{X_1,\ldots,$ $X_s\}$ such that $|\mathcal{X}|=t_i$ and any $s-t_o$ outputs $\mathcal{Y}\subseteq \{Y_1,\ldots,Y_s\}$ such that $|\mathcal{Y}|=s-t_o$, we have 
\begin{align}
     \label{lemma-general-ineq-L}
     H(\mathcal{X}|\mathcal{Y})
     &\ge 
     \sum_{i\in [s]} H(X_i) - (t_o-t_i)\log(v)- H(\mathcal{Y}),\\ \label{lemma-general-ineq-R}
     H(\mathcal{X}|\mathcal{Y}) 
     &\le 
     \min\bigg\{
     \sum_{i\in [s]} H(X_i) - H(\mathcal{Y}),
     (s+t_i-t_o)\log(v)  - H(\mathcal{Y})\bigg\}.
\end{align}
\end{lemma}

The following log-sum inequality will be exploited. 

\begin{lemma}[\cite{CT-book}] \label{lemma-log-sum}
For positive numbers $a_1,a_2,\ldots,a_n$ and $b_1,b_2,\ldots,b_n$, 
\begin{align*}
    \sum_{i=1}^n a_i \log \frac{a_i}{b_i} \ge 
    \bigg(\sum_{i=1}^n a_i\bigg) \log\frac{\sum_{i=1}^n a_i}{\sum_{i=1}^n b_i}
\end{align*}
with equality if and only if $\frac{a_i}{b_i}=constant$. 
\end{lemma}

\begin{proof}[Proof of Lemma~\ref{lemma-Asymmetric-general}] 
Let $\overline{\mathcal{X}} = \{X_1,\ldots,X_s\}\setminus \mathcal{X}$. 
For any $\mathbf{u}\in \Gamma^{t_i},  \mathbf{v}\in \Gamma^{s-t_o}$, denote 
\begin{align}\label{def-eq-U'}
    \mathcal{U}'_{\mathbf{u},\mathbf{v}} := \{\mathbf{u}'\in \Gamma^{s-t_i}: 
    \exists\ \text{a row in array $A$ such that}\ \mathcal{X}=\mathbf{u}, \overline{\mathcal{X}}=\mathbf{u}', \mathcal{Y}=\mathbf{v} \}.
\end{align} 
According to the unbiased property of array $A$, it holds that 
\begin{align}\label{eq-U'-size}
    |\mathcal{U}'_{\mathbf{u},\mathbf{v}}| 
    =v^{t_o-t_i}.
\end{align} 

By the definition, we have 
\begin{align}
    \nonumber
    H(\mathcal{X},\mathcal{Y}) 
    &= -\sum_{\mathbf{u}\in \Gamma^{t_i},  \mathbf{v}\in \Gamma^{s-t_o}} Pr[\mathcal{X}=\mathbf{u},\mathcal{Y}=\mathbf{v}] \log \big(Pr[\mathcal{X}=\mathbf{u},\mathcal{Y}=\mathbf{v}]\big)\\
    \label{eq-y-to-x-bar-Asy}
    &= -\sum_{\mathbf{u}\in \Gamma^{t_i},  \mathbf{v}\in \Gamma^{s-t_o}} \bigg(\sum_{\mathbf{u}'\in \mathcal{U}'_{\mathbf{u},\mathbf{v}}}Pr[\mathcal{X}=\mathbf{u},\overline{\mathcal{X}}=\mathbf{u}'] \bigg)
    \cdot 
    \log \bigg(\sum_{\mathbf{u}'\in \mathcal{U}'_{\mathbf{u},\mathbf{v}}}Pr[\mathcal{X}=\mathbf{u},\overline{\mathcal{X}}=\mathbf{u}'] \bigg)\\
    \nonumber
    & = 
     -\sum_{\mathbf{u}\in \Gamma^{t_i},  \mathbf{v}\in \Gamma^{s-t_o}} \bigg(\sum_{\mathbf{u}'\in \mathcal{U}'_{\mathbf{u},\mathbf{v}}}Pr[\mathcal{X}=\mathbf{u}]\cdot Pr[\overline{\mathcal{X}}=\mathbf{u}'] \bigg)
    \cdot 
    \log \bigg(\sum_{\mathbf{u}'\in \mathcal{U}'_{\mathbf{u},\mathbf{v}}}Pr[\mathcal{X}=\mathbf{u}]\cdot Pr[\overline{\mathcal{X}}=\mathbf{u}'] \bigg)\\
    \label{eq-thm-main2-lemma-divide-2cases-1}
    & = 
     -\sum_{\mathbf{u}\in \Gamma^{t_i},  \mathbf{v}\in \Gamma^{s-t_o}} 
     Pr[\mathcal{X}=\mathbf{u}]
     \bigg(\sum_{\mathbf{u}'\in \mathcal{U}'_{\mathbf{u},\mathbf{v}}}Pr[\overline{\mathcal{X}}=\mathbf{u}'] \bigg)
    \log \bigg(\sum_{\mathbf{u}'\in \mathcal{U}'_{\mathbf{u},\mathbf{v}}}Pr[\overline{\mathcal{X}}=\mathbf{u}'] \bigg)\\
    \label{eq-thm-main2-lemma-divide-2cases-2} 
    & \qquad - 
    \sum_{\mathbf{u}\in \Gamma^{t_i}} 
     Pr[\mathcal{X}=\mathbf{u}] \log \big(Pr[\mathcal{X}=\mathbf{u}]\big)
     \bigg(\sum_{\mathbf{v}\in \Gamma^{s-t_o},\mathbf{u}'\in \mathcal{U}'_{\mathbf{u},\mathbf{v}}}
     Pr[\overline{\mathcal{X}}=\mathbf{u}'] \bigg)
\end{align}
where~\eqref{eq-y-to-x-bar-Asy} follows from the unbiased property (on $\mathcal{X}$ and $\mathcal{Y}$) of the combinatorial asymmetric $(t_i,t_o,s,v)$-AONT. Next we estimate the two items \eqref{eq-thm-main2-lemma-divide-2cases-1} and  \eqref{eq-thm-main2-lemma-divide-2cases-2} as follows. 
\begin{align}
    \nonumber 
    \eqref{eq-thm-main2-lemma-divide-2cases-1} 
    \nonumber 
    & = 
    -\sum_{\mathbf{u}\in \Gamma^{t_i},  \mathbf{v}\in \Gamma^{s-t_o}} 
     Pr[\mathcal{X}=\mathbf{u}]\cdot 
     |\mathcal{U}'_{\mathbf{u},\mathbf{v}}|
     \cdot  
     \bigg(\sum_{\mathbf{u}'\in \mathcal{U}'_{\mathbf{u},\mathbf{v}}}
     \frac{1}{|\mathcal{U}'_{\mathbf{u},\mathbf{v}}| } Pr[\overline{\mathcal{X}}=\mathbf{u}'] \bigg)
    \log \bigg(\sum_{\mathbf{u}'\in \mathcal{U}'_{\mathbf{u},\mathbf{v}}}
    \frac{1}{|\mathcal{U}'_{\mathbf{u},\mathbf{v}}| }
    Pr[\overline{\mathcal{X}}=\mathbf{u}'] \bigg) \\
    \nonumber 
    & \qquad 
    - \sum_{\mathbf{u}\in \Gamma^{t_i},  \mathbf{v}\in \Gamma^{s-t_o}} 
     Pr[\mathcal{X}=\mathbf{u}]
     \bigg(\sum_{\mathbf{u}'\in \mathcal{U}'_{\mathbf{u},\mathbf{v}}}Pr[\overline{\mathcal{X}}=\mathbf{u}'] \bigg)
    \log \Big( |\mathcal{U}'_{\mathbf{u},\mathbf{v}}| \Big)\\
    \label{eq-lemma-first-case-subcase-1}
    & \ge 
    -\sum_{\mathbf{u}\in \Gamma^{t_i},  \mathbf{v}\in \Gamma^{s-t_o}} 
     Pr[\mathcal{X}=\mathbf{u}] \cdot 
     |\mathcal{U}'_{\mathbf{u},\mathbf{v}}|
     \cdot  
     \bigg(\sum_{\mathbf{u}'\in \mathcal{U}'_{\mathbf{u},\mathbf{v}}}
     \frac{1}{|\mathcal{U}'_{\mathbf{u},\mathbf{v}}| } Pr[\overline{\mathcal{X}}=\mathbf{u}']
    \log \big(
    Pr[\overline{\mathcal{X}}=\mathbf{u}'] \big) \bigg)\\
    \label{eq-lemma-first-case-subcase-2}
    & \qquad 
    - (t_0-t_i)\log(v) \sum_{\mathbf{u}\in \Gamma^{t_i},  \mathbf{u}'\in \Gamma^{s-t_i}} 
     Pr[\mathcal{X}=\mathbf{u}]Pr[\overline{\mathcal{X}}=\mathbf{u}']\\ 
    \label{eq-lemma-first-case-subcase-1-1}
    & = 
    -\sum_{\mathbf{u}\in \Gamma^{t_i}} 
     Pr[\mathcal{X}=\mathbf{u}] 
     \sum_{\mathbf{u}'\in \Gamma^{s-t_i}}
      Pr[\overline{\mathcal{X}}=\mathbf{u}']
    \log \big(
    Pr[\overline{\mathcal{X}}=\mathbf{u}'] \big) - (t_0-t_i)\log(v)\\
    \label{eq-lemma-first-case-1-final} 
    & = 
    H(\overline{\mathcal{X}})- (t_0-t_i)\log(v)
\end{align}
where \eqref{eq-lemma-first-case-subcase-1} follows from Lemma~\ref{lemma-Jensen-ineq} with $\sum\limits_{\mathbf{u}'\in \mathcal{U}'_{\mathbf{u},\mathbf{v}}}
\frac{1}{|\mathcal{U}'_{\mathbf{u},\mathbf{v}}| }=1$ and $Pr[\overline{\mathcal{X}}=\mathbf{u}']\ge 0$; \eqref{eq-lemma-first-case-subcase-2} and \eqref{eq-lemma-first-case-subcase-1-1} follow from \eqref{def-eq-U'} and \eqref{eq-U'-size}. Also 
\begin{align}
    \label{eq-lemma-first-case-2-final} 
    \eqref{eq-thm-main2-lemma-divide-2cases-2} 
    & = 
    - \sum_{\mathbf{u}\in \Gamma^{t_i}} 
     Pr[\mathcal{X}=\mathbf{u}] \log \big(Pr[\mathcal{X}=\mathbf{u}]\big)
     \bigg(\sum_{\mathbf{u}'\in \Gamma^{s-t_i}}
     Pr[\overline{\mathcal{X}}=\mathbf{u}'] \bigg)
     = H(\mathcal{X}). 
\end{align}
Plugging \eqref{eq-lemma-first-case-1-final} and \eqref{eq-lemma-first-case-2-final} into \eqref{eq-thm-main2-lemma-divide-2cases-1} and \eqref{eq-thm-main2-lemma-divide-2cases-2} yields 
\begin{align*}
    H(\mathcal{X},\mathcal{Y})  \ge 
    H(\overline{\mathcal{X}})- (t_0-t_i)\log(v) + H(\mathcal{X})
    = 
    \sum_{i\in [s]} H(X_i) - (t_o-t_i)\log(v)
\end{align*}
and hence 
\begin{align*}
     H(\mathcal{X}|\mathcal{Y}) 
     = H(\mathcal{X},\mathcal{Y})  - H(\mathcal{Y}) 
     \ge 
     \sum_{i\in [s]} H(X_i) - (t_o-t_i)\log(v) - H(\mathcal{Y}),
\end{align*}
implying the lower bound \eqref{lemma-general-ineq-L}.

Next we verify the upper bound \eqref{lemma-general-ineq-R}. 
First, by \eqref{eq-y-to-x-bar-Asy} and the 
monotonically increasing property of $\log_2(\cdot)$ function, we have 
\begin{align*}
    H(\mathcal{X},\mathcal{Y})
    &\le 
    -\sum_{\mathbf{u}\in \Gamma^{t_i},  \mathbf{v}\in \Gamma^{s-t_o}} 
    \sum_{\mathbf{u}'\in \mathcal{U}'_{\mathbf{u},\mathbf{v}}}Pr[\mathcal{X}=\mathbf{u},\overline{\mathcal{X}}=\mathbf{u}'] 
    \log \Big(Pr[\mathcal{X}=\mathbf{u},\overline{\mathcal{X}}=\mathbf{u}'] \Big)\\
    & = 
        -\sum_{\mathbf{u}\in \Gamma^{t_i},  \mathbf{u}'\in \Gamma^{s-t_i}} Pr[\mathcal{X}=\mathbf{u},\overline{\mathcal{X}}=\mathbf{u}'] 
    \log \Big(Pr[\mathcal{X}=\mathbf{u},\overline{\mathcal{X}}=\mathbf{u}'] \Big)\\
    & = H(\mathcal{X},\overline{\mathcal{X}}), 
\end{align*}
implying 
\begin{align}\label{eq-lemma-asy-general-upper-1}
     H(\mathcal{X}|\mathcal{Y}) 
     \le H(\mathcal{X},\overline{\mathcal{X}})  - H(\mathcal{Y}) 
     = 
     \sum_{i\in [s]} H(X_i) - H(\mathcal{Y}).
\end{align}
Also, from \eqref{eq-y-to-x-bar-Asy} together with Lemma~\ref{lemma-log-sum} we have 
\begin{align*}
    H(\mathcal{X},\mathcal{Y})  
    &\le 
    -\sum_{\mathbf{u}\in \Gamma^{t_i},  \mathbf{v}\in \Gamma^{s-t_o}} 
    \sum_{\mathbf{u}'\in \mathcal{U}'_{\mathbf{u},\mathbf{v}}}Pr[\mathcal{X}=\mathbf{u},\overline{\mathcal{X}}=\mathbf{u}'] 
   \cdot 
    \log 
    \frac{\sum\limits_{\mathbf{u}\in \Gamma^{t_i},  \mathbf{v}\in \Gamma^{s-t_o}} 
    \sum\limits_{\mathbf{u}'\in \mathcal{U}'_{\mathbf{u},\mathbf{v}}}Pr[\mathcal{X}=\mathbf{u},\overline{\mathcal{X}}=\mathbf{u}'] }{\sum\limits_{\mathbf{u}\in \Gamma^{t_i},  \mathbf{v}\in \Gamma^{s-t_o}} 1}\\
    & = 
        (s+t_i-t_o)\log(v)
\end{align*}
which implies  
\begin{align}
\label{eq-lemma-asy-general-upper-2}
     H(\mathcal{X}|\mathcal{Y}) 
     \le (s+t_i-t_o)\log(v)  - H(\mathcal{Y}).
\end{align}
Combining \eqref{eq-lemma-asy-general-upper-1} and \eqref{eq-lemma-asy-general-upper-2} gives the upper bound.  
This completes the proof.
\end{proof}

It is worth noting that when $t_i=t_o$, the above Lemma~\ref{lemma-Asymmetric-general} implies Lemma~\ref{lemma-symmetric-general}. In other words, the lower and upper bounds in Lemma~\ref{lemma-Asymmetric-general} turn out to be tight in certain cases.

\subsection{Proof of Theorem~\ref{thm-main-Asymmetric}} 

\begin{proof}[Proof of Theorem~\ref{thm-main-Asymmetric}]

(1)
According to Lemma~\ref{lemma-Asymmetric-general}, we have 
\begin{align*}
    H(\mathcal{X}|\mathcal{Y})  
    &\ge 
    \sum_{i\in [s]} H(X_i) - (t_o-t_i)\log(v)- H(\mathcal{Y})\\ &\ge 
    \sum_{i\in [s]} H(X_i) - (t_o-t_i)\log(v)- (s-t_o)\log(v) \\
   & = 
    \sum_{i\in [s]} H(X_i) - (s-t_i)\log(v)
\end{align*}
where the second inequality follows from the fact $H(Y,Y')\le H(Y)+H(Y')$ and $H(Y)\le \log(v)$ for any output $Y$. 
Together with the non-negativity of entropy, the lower bound \eqref{eq-thm-Asy-main-1-ld} follows.

(2)
Recall that $X_1,\ldots,X_s$ are mutually independent. 
Let $\mathcal{X}_{\max}\subseteq \{X_1,\ldots,X_s\}$ such that $|\mathcal{X}_{\max}|=s-t_i$ denote a collection of input random variables according to the largest $s-t_i$ entropy values among $H(X_1),\ldots,H(X_s)$.  
Denote $\overline{\mathcal{X}_{\max}} = \{X_1,\ldots,X_s\} \setminus \mathcal{X}_{\max}$. Clearly, $|\overline{\mathcal{X}_{\max}}|=t_i$.
In order to derive an upper bound on $H(\mathcal{X}|\mathcal{Y})$ based on Lemma~\ref{lemma-Asymmetric-general}, we need to evaluate on $H(\mathcal{Y})$. 
\begin{align}
    \nonumber 
    &
    H(\mathcal{Y}) \\
    \nonumber 
    & = -\sum_{\mathbf{v}\in \Gamma^{s-t_o}} Pr[\mathcal{Y}=\mathbf{v}] \log  \big(Pr[\mathcal{Y}=\mathbf{v}]\big)\\
    \label{eq-thm5.2-proof-use-CANOT-1}
    & = -\sum_{\mathbf{v}\in \Gamma^{s-t_o}} 
       \Bigg( \sum_{\mathbf{u}\in \Gamma^{t_i}} \sum_{\mathbf{u}'\in \mathcal{U}'_{\mathbf{u},\mathbf{v}} } Pr[\overline{\mathcal{X}_{\max}}=\mathbf{u},\mathcal{X}_{\max}=\mathbf{u}'] \Bigg) 
      \cdot 
       \log 
       \Bigg( \sum_{\mathbf{u}\in \Gamma^{t_i}} \sum_{\mathbf{u}'\in \mathcal{U}'_{\mathbf{u},\mathbf{v}} } Pr[\overline{\mathcal{X}_{\max}}=\mathbf{u},\mathcal{X}_{\max}=\mathbf{u}'] \Bigg) \\
     \nonumber 
     &
     = 
       -\sum_{\mathbf{v}\in \Gamma^{s-t_o}} 
       \Bigg( 
       \sum_{\mathbf{u}\in \Gamma^{t_i}} 
       Pr[\overline{\mathcal{X}_{\max}}=\mathbf{u}]
       \sum_{\mathbf{u}'\in \mathcal{U}'_{\mathbf{u},\mathbf{v}} } Pr[\mathcal{X}_{\max}=\mathbf{u}'] \Bigg) 
    \cdot 
       \log 
       \Bigg( 
       \sum_{\mathbf{u}\in \Gamma^{t_i}} 
       Pr[\overline{\mathcal{X}_{\max}}=\mathbf{u}]
       \sum_{\mathbf{u}'\in \mathcal{U}'_{\mathbf{u},\mathbf{v}} } Pr[\mathcal{X}_{\max}=\mathbf{u}'] \Bigg) \\
    \label{eq-proof-thm-main-2-use-Jensen-1}
    &  
    \ge 
        -\sum_{\mathbf{v}\in \Gamma^{s-t_o}} 
       \Bigg( 
       \sum_{\mathbf{u}\in \Gamma^{t_i}} 
       Pr[\overline{\mathcal{X}_{\max}}=\mathbf{u}]
    \cdot        
       \bigg( 
       \sum_{\mathbf{u}'\in \mathcal{U}'_{\mathbf{u},\mathbf{v}} } Pr[\mathcal{X}_{\max}=\mathbf{u}'] 
       \bigg)
       \log 
       \bigg( 
       \sum_{\mathbf{u}'\in \mathcal{U}'_{\mathbf{u},\mathbf{v}} } Pr[\mathcal{X}_{\max}=\mathbf{u}']\bigg)  \Bigg) \\
    \nonumber 
    & 
    = 
    -\sum_{\mathbf{u}\in \Gamma^{t_i},  \mathbf{v}\in \Gamma^{s-t_o}} 
     Pr[\overline{\mathcal{X}_{\max}}=\mathbf{u}]\cdot 
     |\mathcal{U}'_{\mathbf{u},\mathbf{v}}|
     \cdot  
     \bigg(\sum_{\mathbf{u}'\in \mathcal{U}'_{\mathbf{u},\mathbf{v}}}
     \frac{1}{|\mathcal{U}'_{\mathbf{u},\mathbf{v}}| } Pr[\mathcal{X}_{\max}=\mathbf{u}'] \bigg)
    \log \bigg(\sum_{\mathbf{u}'\in \mathcal{U}'_{\mathbf{u},\mathbf{v}}}
    \frac{1}{|\mathcal{U}'_{\mathbf{u},\mathbf{v}}| }
    Pr[\mathcal{X}_{\max}=\mathbf{u}'] \bigg) \\
    \nonumber 
    & \qquad 
    - \sum_{\mathbf{u}\in \Gamma^{t_i},  \mathbf{v}\in \Gamma^{s-t_o}} 
     Pr[\overline{\mathcal{X}_{\max}}=\mathbf{u}]
     \bigg(\sum_{\mathbf{u}'\in \mathcal{U}'_{\mathbf{u},\mathbf{v}}}
     Pr[\mathcal{X}_{\max}=\mathbf{u}'] \bigg)
    \log \Big( |\mathcal{U}'_{\mathbf{u},\mathbf{v}}| \Big)\\
    \label{eq-proof-thm-main-2-use-Jensen-2}
    &
    \ge 
    -\sum_{\mathbf{u}\in \Gamma^{t_i},  \mathbf{v}\in \Gamma^{s-t_o}} 
     Pr[\overline{\mathcal{X}_{\max}}=\mathbf{u}]\cdot 
     |\mathcal{U}'_{\mathbf{u},\mathbf{v}}|
    \cdot  
     \bigg(\sum_{\mathbf{u}'\in \mathcal{U}'_{\mathbf{u},\mathbf{v}}}
     \frac{1}{|\mathcal{U}'_{\mathbf{u},\mathbf{v}}| } Pr[\mathcal{X}_{\max}=\mathbf{u}'] 
    \log \Big(
    Pr[\mathcal{X}_{\max}=\mathbf{u}'] \Big) \bigg) \\
    \nonumber 
    & \qquad 
    - (t_o-t_i)\log(v)\\
    \nonumber 
    & 
    = 
    -\sum_{\mathbf{u}\in \Gamma^{t_i}} 
     Pr[\overline{\mathcal{X}_{\max}}=\mathbf{u}]  
     \sum_{\mathbf{u}'\in \Gamma^{s-t_o}}
     Pr[\mathcal{X}_{\max}=\mathbf{u}'] 
     \log \Big(
     Pr[\mathcal{X}_{\max}=\mathbf{u}'] \Big)  
     - (t_o-t_i)\log(v)\\
    \label{eq-thm-main2-proof-2-last}
    & 
    = 
    H(\mathcal{X}_{\max}) - (t_o-t_i)\log(v)
\end{align}
where \eqref{eq-thm5.2-proof-use-CANOT-1} follows from the assumption of combinatorial asymmetric $(t_i,t_o,s,$ $v)$-AONT and \eqref{def-eq-U'}; 
\eqref{eq-proof-thm-main-2-use-Jensen-1} follows from Lemma~\ref{lemma-Jensen-ineq} with $\sum\limits_{\mathbf{u}\in \Gamma^{t_i}} 
Pr[\overline{\mathcal{X}_{\max}}=\mathbf{u}]=1$ and $\sum\limits_{\mathbf{u}'\in \mathcal{U}'_{\mathbf{u},\mathbf{v}} } Pr[\mathcal{X}_{\max}=\mathbf{u}'] \ge 0$; 
and \eqref{eq-proof-thm-main-2-use-Jensen-2} follows from Lemma~\ref{lemma-Jensen-ineq} with 
$\sum\limits_{\mathbf{u}'\in \mathcal{U}'_{\mathbf{u},\mathbf{v}}}
     \frac{1}{|\mathcal{U}'_{\mathbf{u},\mathbf{v}}| } = 1$ and  
$Pr[\mathcal{X}_{\max}=\mathbf{u}']\ge 0$.

Combining Lemma~\ref{lemma-Asymmetric-general} and \eqref{eq-thm-main2-proof-2-last} yields 
\begin{align*}
    H(\mathcal{X}|\mathcal{Y}) 
    &\le 
    \sum_{i\in [s]} H(X_i) - H(\mathcal{Y}) 
    \le 
    \sum_{i\in [s]} H(X_i) -  H(\mathcal{X}_{\max}) + (t_o-t_i)\log(v) \\
    &= 
    \min_{I\subseteq [s],\atop|I|=t_i}\, \sum_{i\in I} H(X_i) + (t_o-t_i)\log(v),
\end{align*}
as well as 
\begin{align*}
    H(\mathcal{X}|\mathcal{Y}) 
    &\le 
    (s+t_i-t_o)\log(v) - H(\mathcal{Y}) \\
    &\le 
    (s+t_i-t_o)\log(v) -  H(\mathcal{X}_{\max}) + (t_o-t_i)\log(v) \\
    &= 
    \min_{I\subseteq [s],\atop|I|=t_i}\, \sum_{i\in I} H(X_i) + s\cdot\log(v)-\sum^s_{i=1}H(X_i), 
\end{align*}
as desired. 
This completes the proof. 
\end{proof}

We remark that the lower and upper bounds of Theorem~\ref{thm-main-Asymmetric} coincide 
when all $s$ inputs are with the uniform distribution and accordingly the combinatorial asymmetric $(t_i,t_o,s,v)$-AONT has perfect security (see also \cite[Theorem 2.3]{ES2021+}). 
However, unlike the bounds in Theorem~\ref{thm-main-symmetric},  $H(\mathcal{X}|\mathcal{Y})$ for combinatorial  asymmetric $(t_i,t_o,s,v)$-AONTs  cannot be bounded above by $\min\limits_{\mathcal{X}'}H(\mathcal{X}')$ in general (see Example~\ref{example-asy}). 

\begin{example}\label{example-asy}
Consider the following $(1,2,3,3)$-AONT over the alphabet $\Gamma=\{a,b,c\}$ shown in Table~\ref{tableII}.\\

\begin{table}[h]
 \centering
  \begin{tabular}{|c|c|c||c|c|c|}
   \hline
   $X_1$ & $X_2$ & $X_3$ & $Y_1$ & $Y_2$ & $Y_3$\\
   \hline
   $a$ & $a$ & $a$ & $a$ & $a$ & $a$\\
   $a$ & $a$ & $b$ & $b$ & $b$ & $a$\\
   $a$ & $a$ & $c$ & $c$ & $c$ & $a$\\
   $a$ & $b$ & $a$ & $a$ & $b$ & $b$\\
   $a$ & $b$ & $b$ & $b$ & $c$ & $b$\\
   $a$ & $b$ & $c$ & $c$ & $a$ & $b$\\
   $a$ & $c$ & $a$ & $a$ & $c$ & $c$\\
   $a$ & $c$ & $b$ & $b$ & $a$ & $c$\\
   $a$ & $c$ & $c$ & $c$ & $b$ & $c$\\
   $b$ & $a$ & $a$ & $b$ & $a$ & $b$\\
   $b$ & $a$ & $b$ & $c$ & $b$ & $b$\\
   $b$ & $a$ & $c$ & $a$ & $c$ & $b$\\
   $b$ & $b$ & $a$ & $b$ & $b$ & $c$\\
   $b$ & $b$ & $b$ & $c$ & $c$ & $c$\\
   $b$ & $b$ & $c$ & $a$ & $a$ & $c$\\
   $b$ & $c$ & $a$ & $b$ & $c$ & $a$\\
   $b$ & $c$ & $b$ & $c$ & $a$ & $a$\\
   $b$ & $c$ & $c$ & $a$ & $b$ & $a$\\
   $c$ & $a$ & $a$ & $c$ & $a$ & $c$\\
   $c$ & $a$ & $b$ & $a$ & $b$ & $c$\\
   $c$ & $a$ & $c$ & $b$ & $c$ & $c$\\
   $c$ & $b$ & $a$ & $c$ & $b$ & $a$\\
   $c$ & $b$ & $b$ & $a$ & $c$ & $a$\\
   $c$ & $b$ & $c$ & $b$ & $a$ & $a$\\
   $c$ & $c$ & $a$ & $c$ & $c$ & $b$\\
   $c$ & $c$ & $b$ & $a$ & $a$ & $b$\\
   $c$ & $c$ & $c$ & $b$ & $b$ & $b$\\
   \hline
  \end{tabular}
  \vskip 0.3cm
  \caption{A combinatorial $(1,2,3,3)$-AONT over the alphabet $\{a, b, c\}$}
  \label{tableII}
\end{table}

Suppose that
\begin{align*}
    &Pr[X_1=a]=\frac{1}{6},\  Pr[X_1=b]=\frac{1}{3},\  Pr[X_1=c]=\frac{1}{2},\\
    &Pr[X_2=a]=\frac{1}{2},\  Pr[X_2=b]=\frac{1}{4},\  Pr[X_2=c]=\frac{1}{4},\\
    &Pr[X_3=a]=\frac{7}{10},\  Pr[X_3=b]=\frac{1}{5},\  Pr[X_3=c]=\frac{1}{10}
\end{align*}
which gives 
\begin{align*}
    H(X_1) = 1.459148,\   
    H(X_2) = 1.500000,\   
    H(X_3) = 1.156780. 
\end{align*}
Together with Table~\ref{tableII}, we obtain
\begin{align*}
    H(X_1|Y_1)=1.067794,\  H(X_1|Y_2)=1.459148,\  H(X_1|Y_3)=1.381719,\\
    H(X_2|Y_1)=1.500000,\  H(X_2|Y_2)=1.098856,\  H(X_2|Y_3)=1.381719,\\
    H(X_3|Y_1)=1.067794,\  H(X_3|Y_2)=1.098856,\  H(X_3|Y_3)=1.156780
\end{align*}
from which it is easy to see
    $H(X_i|Y_j) \neq H(X_{i}|Y_{j'} )$   for all $i, j, j'\in \{1, 2, 3\}$ with $j\neq j'$. 
It is readily seen that 
    \begin{align*}
         H(X_1|Y_2) > H(X_3) = \min \{H(X_1), H(X_2), H(X_3)\}.
    \end{align*}
\end{example}

\section{Asymmetric weak-AONT with independent inputs}
\label{sec-Asym-weak-AONT}

This section 
discusses the security properties of combinatorial asymmetric weak-AONT with independent inputs. 
We have the following theorem, which extends the discussions in the preceding Section~\ref{sec-Asym-AONT}.

\begin{theorem}\label{thm-main-Asymmetric-weak} 
Let array $A\in \Gamma^{v^s\times 2s}$ be a combinatorial asymmetric $(t_i,t_o,s,v)$-weak-AONT 
whose columns are with respect to random variables $X_1,\ldots,X_s,Y_1,$ $\ldots,Y_s$ respectively. Let $P_1,\ldots,P_s$ be the corresponding probability distributions of $X_1,\ldots,X_s$, which are  
mutually independent and take values from $\Gamma$. 
Then for any $t_i$ inputs $\mathcal{X}\subseteq \{X_1,\ldots,X_s\}$ such that $|\mathcal{X}|=t_i$ and any $s-t_o$ outputs $\mathcal{Y}\subseteq \{Y_1,\ldots,Y_s\}$ such that $|\mathcal{Y}|=s-t_o$, the followings hold.  
\begin{enumerate}
    \item[(1)] \begin{align}\label{eq-thm-Asy-main-1}
        H(\mathcal{X}|\mathcal{Y})  \ge \max\bigg\{0,\sum_{j=1}^s H(X_j) - (s-t_o)\log(v) -
        \log\big(v^{s-t_i}-v^{s-t_o}+1\big)
        \bigg\}. 
    \end{align}
    \item[(2)] 
    \begin{align}\label{eq-thm-Asy-main-2}
        H(\mathcal{X}|\mathcal{Y}) 
        \le \min\bigg\{H(\mathcal{X}) , \min_{I\subseteq [s],\atop|I|=t_i}\, \sum_{j\in I} H(X_j) + 
        \log\big(v^{s-t_i}-v^{s-t_o}+1\big)
        \bigg\}.
    \end{align}
\end{enumerate}
\end{theorem}

To prove Theorem~\ref{thm-main-Asymmetric-weak}, we will make use of the following lemma. 

\begin{lemma}\label{lemma-Asymmetric-general-weak}
Under the assumption of Theorem~\ref{thm-main-Asymmetric-weak}, for any $t_i$ inputs $\mathcal{X}\subseteq \{X_1,\ldots,$ $X_s\}$ such that $|\mathcal{X}|=t_i$ and any $s-t_o$ outputs $\mathcal{Y}\subseteq \{Y_1,\ldots,Y_s\}$ such that $|\mathcal{Y}|=s-t_o$, we have 
\begin{align}\label{lemma-general-ineq-asy-weak}
     \sum_{j\in [s]} H(X_j) - \log\big(v^{s-t_i}-v^{s-t_o}+1\big)- H(\mathcal{Y}) \le H(\mathcal{X}|\mathcal{Y}) \le \sum_{j\in [s]} H(X_j) - H(\mathcal{Y}).
\end{align}
\end{lemma}

\begin{proof}
The proof can be done by following the same line as the arguments for Lemma~\ref{lemma-Asymmetric-general}. 
The only difference is as follows. 
Recall from \eqref{def-eq-U'} that 
\begin{align*}
    \mathcal{U}'_{\mathbf{u},\mathbf{v}} 
    = \{\mathbf{u}'\in \Gamma^{s-t_i}: 
    \exists\ \text{a row in array $A$ such that}\ \mathcal{X}=\mathbf{u}, \overline{\mathcal{X}}=\mathbf{u}', \mathcal{Y}=\mathbf{v} \}
\end{align*} 
for any $\mathbf{u}\in \Gamma^{t_i},  \mathbf{v}\in \Gamma^{s-t_o}$, and it holds that $|\mathcal{U}'_{\mathbf{u},\mathbf{v}} |=v^{t_0-t_i}$ for a combinatorial asymmetric $(t_i,t_o,s,v)$-AONT due to its unbiased property. 
However, for a combinatorial asymmetric $(t_i,t_o,s,v)$-weak-AONT, we only have 
\begin{align}\label{eq-size-U'-asy-weak} 
  1\le |\mathcal{U}'_{\mathbf{u},\mathbf{v}} | \le v^{s-t_i}-v^{s-t_o} +1
\end{align}
according to its covering property. 
Replacing the quantization on $ |\mathcal{U}'_{\mathbf{u},\mathbf{v}} |$ in the proof of Lemma~\ref{lemma-Asymmetric-general} by the above estimation \eqref{eq-size-U'-asy-weak}, it is not hard to derive the inequality \eqref{lemma-general-ineq-asy-weak}, and hence the lemma follows. 
\end{proof}


\begin{proof}[Proof of Theorem~\ref{thm-main-Asymmetric-weak}]
Based on Lemma~\ref{lemma-Asymmetric-general-weak}, the proof follows the same line as the argument of Theorem~\ref{thm-main-Asymmetric} by modifying the quantization of $ |\mathcal{U}'_{\mathbf{u},\mathbf{v}} |$  by the inequality \eqref{eq-size-U'-asy-weak}.  
\end{proof}

We would remark that the bounds in Theorem~\ref{thm-main-Asymmetric-weak} could be tight in special cases, say when $t_i=t_o$ and all the inputs are with uniform distribution. However, it is not always tight, see Example~\ref{example-asy-weak-4} below. 
Also in contrast to Theorem~\ref{thm-main-symmetric}, Example~\ref{example-asy-weak-4} shows that neither the relation $H(\mathcal{X}|\mathcal{Y})=H(\mathcal{X}|\mathcal{Y}')$ nor $H(\mathcal{X}|\mathcal{Y})=H(\mathcal{X}'|\mathcal{Y})$ holds in general for any distinct $\mathcal{X},\mathcal{X}'$ and any distinct $\mathcal{Y},\mathcal{Y}'$.

\begin{example}\label{example-asy-weak-4}\rm 
Consider the following $(1,2,3, 2)$-weak-AONT over the alphabet $\Gamma=\{a, b\}$ shown in Table~\ref{table-example-asy-weak} as in \cite{ES2021+asy}. 

\begin{table}[h]
 \centering
  \begin{tabular}{|c|c|c||c|c|c|}
  \hline
  $X_1$ & $X_2$ & $X_3$ & $Y_1$ & $Y_2$ & $Y_3$\\
  \hline
  $a$ & $a$ & $a$ & $a$ & $a$ & $a$\\
  $a$ & $a$ & $b$ & $b$ & $b$ & $a$\\
  $a$ & $b$ & $a$ & $b$ & $a$ & $b$\\
  $a$ & $b$ & $b$ & $b$ & $a$ & $a$\\
  $b$ & $a$ & $a$ & $a$ & $b$ & $b$\\
  $b$ & $a$ & $b$ & $a$ & $b$ & $a$\\
  $b$ & $b$ & $a$ & $a$ & $a$ & $b$\\
  $b$ & $b$ & $b$ & $b$ & $b$ & $b$\\
  \hline
  \end{tabular}
  \vskip 0.3cm
  \caption{A $(1,2,3,2)$-weak-AONT over the alphabet $\{a, b\}$}
  \label{table-example-asy-weak}
\end{table}
Suppose that
\begin{align*}
    Pr[X_1=a]=\frac{1}{4},\  Pr[X_1=b]=\frac{3}{4},\\  Pr[X_2=a]=\frac{1}{3},\  Pr[X_2=b]=\frac{2}{3},\\
    Pr[X_3=a]=\frac{1}{2},\  Pr[X_3=b]=\frac{1}{2}.
\end{align*}
Plugging into Table~\ref{table-example-asy-weak} gives 
\begin{align*}
    Pr[Y_1=a]=\frac{13}{24},\  Pr[Y_1=b]=\frac{11}{24},\\  Pr[Y_2=a]=\frac{11}{24},\
    Pr[Y_2=b]=\frac{13}{24},\\  Pr[Y_3=a]=\frac{7}{24},\  Pr[Y_3=b]=\frac{17}{24}.
\end{align*}
Then by the definition, it is easy to calculate 
\begin{align*}
    H(X_1)=0.811278, \ H(X_2)=0.918296,\ H(X_3)=1.00000,\\ 
    H(Y_1)=0.994985, \ H(Y_2)=0.994985,\ H(Y_3)=0.870864.
\end{align*}
Furthermore 
\begin{align*}
    H(X_1|Y_1)=0.667521, \ H(X_1|Y_2)=0.667521, \ H(X_1|Y_3)=0.657504,\\
    H(X_2|Y_1)=0.740788, \ H(X_2|Y_2)=0.740788, \ H(X_2|Y_3)=0.727952,\\
    H(X_3|Y_1)=0.735665, \ H(X_3|Y_2)=0.735665, \  H(X_3|Y_3)=0.836044, 
\end{align*}
from which it is easy to see $H(X_1|Y_1)=H(X_1|Y_2) \neq H(X_1|Y_3)$ and
$H(X_1|Y_1) \neq H(X_2|Y_1)$. 
\end{example}


\section{Conclusion}
\label{sec-conclusion} 
In this paper, we initially investigated the  
security properties sandwiched between perfect security and weak security for combinatorial 
AONT and combinatorial asymmetric AONT in the scenarios that 
all the $s$ inputs take values independently but not necessarily identically and the even less restrictive model allowing partial dependency.  
By using information-theoretic techniques, we established general lower and upper bounds on the amount of information $H(\mathcal{X}|\mathcal{Y})$ about any $t_i$ inputs $\mathcal{X}$ that is not revealed by any $s-t_o$ outputs $\mathcal{Y}$.  
It is also proven that the derived bounds could be attained in certain cases. 
However the security properties of combinatorial (asymmetric) AONT are still unknown for many non-independent and non-identical (prior) probability distributions on the inputs, which is indeed of interest and worth investigating in the future work.
In addition,  
to investigate the information-theoretic security properties of \textit{linear} AONTs~\cite{Stinson2001}, in which each of the outputs is a linear combination of inputs and could be computed efficiently, with some prior input distributions
is an interesting direction as well.

\section*{Acknowledgment}
The authors would like to thank Professor Douglas R. Stinson for reading an early version of this manuscript.


\end{document}